\documentclass[10pt,a4paper,english]{amsart}
\usepackage{amssymb, amsthm}

\newcommand{\cA}{\mathcal{A}}

\newcommand{\cF}{\mathcal{F}}

\newcommand{\cE}{\mathcal{E}}

\newcommand{\cN}{\mathcal{N}}

\newcommand{\al}{\alpha}

\newcommand{\ga}{\gamma}
\newcommand{\del}{\delta}
\newcommand{\ka}{\kappa}
\newcommand{\Om}{\Omega}
\newcommand{\si}{\sigma}
\newcommand{\Si}{\Sigma}

\newcommand{\R}{\mathbb{R}}
\newcommand{\C}{\mathbb{C}}

\newcommand{\E}{\mathbb{E}}

\newcommand{\bP}{\mathbb{P}}

\newcommand{\la}{\lambda}
\newcommand{\eps}{\varepsilon}
\newcommand{\ti}{\times}

\newcommand{\norm}[1]{\left\|#1\right\|}%

\newcommand{\Prob}{\mathbb{P}}

\newcommand{\bZ}{{\mathbf Z}}

\newcommand{\xspace}{\hbox{\kern-2.5pt}}

\newtheorem{theorem}{Theorem}[section]
\newtheorem{lemma}[theorem]{Lemma}
\newtheorem{corollary}[theorem]{Corollary}

\theoremstyle{definition}

\newtheorem{design-crit}[theorem]{Design Criterion}
\newtheorem{remark}[theorem]{Remark}

\theoremstyle{remark}
\newtheorem{example}[theorem]{Example}

\begin{document}

\title{On the gap between RIP-properties and sparse recovery conditions}

\author{Sjoerd Dirksen}

\address{RWTH Aachen University\\
Lehrstuhl C f{\"u}r Mathematik (Analysis)\\
Pontdriesch 10\\
52062 Aachen\\
Germany} \email{dirksen@mathc.rwth-aachen.de}

\author{Guillaume Lecu\'{e}}

\address{CNRS and Ecole Polytechnique \\
CMAP\\
Route de Saclay\\
91120 Palaiseau \\
France} \email{guillaume.lecue@cmap.polytechnique.fr}

\author{Holger Rauhut}

\address{RWTH Aachen University\\
Lehrstuhl C f{\"u}r Mathematik (Analysis)\\
Pontdriesch 10\\
52062 Aachen\\
Germany} \email{rauhut@mathc.rwth-aachen.de}

\keywords{Restricted isometry property, compressive sensing, $\ell_p$-constrained basis pursuit, Gaussian random matrix, quantized compressive sensing.}

\maketitle

\begin{abstract}
We consider the problem of recovering sparse vectors from underdetermined linear measurements via $\ell_p$-constrained basis pursuit. Previous analyses of this problem based on generalized restricted isometry properties have suggested that two phenomena occur if $p\neq 2$. First, one may need substantially more than $s \log(en/s)$ measurements (optimal for $p=2$) for uniform recovery of all $s$-sparse vectors. Second, the matrix that achieves recovery with the optimal number of measurements may not be Gaussian (as for $p=2$). We present a new, direct analysis which shows that in fact neither of these phenomena occur. Via a suitable version of the null space property we show that a standard Gaussian matrix provides $\ell_q/\ell_1$-recovery guarantees for $\ell_p$-constrained basis pursuit in the optimal measurement regime. Our result extends to several heavier-tailed measurement matrices. As an application, we show that one can obtain a consistent reconstruction from uniform scalar quantized measurements in the optimal measurement regime.  
\end{abstract}

\section{Introduction}

Compressive sensing \cite{do06-2,carota06,FoR13} 
has established itself in the recent years as a rapidly growing research area with various promising 
signal and image processing applications and beyond, and which has triggered many developments on the theoretical side. 
The theory predicts that (approximately) sparse signals can be accurately recovered from incomplete and perturbed linear measurements. 
The measurement process is described by a measurement matrix $A \in \C^{m \times n}$ with $m < n$. While the na\"{i}ve reconstruction 
approach via $\ell_0$-minimization is NP-hard \cite{Nat95}, several tractable recovery methods have been proposed including basis pursuit 
($\ell_1$-minimization), iterative hard thresholding and greedy methods. For all these methods rigorous recovery guarantees have been shown,
see \cite{FoR13} for details and further references.

The restricted isometry property (RIP) is a well-established tool to analyze the performance of sparse recovery methods. 
The standard version defines the restricted isometry constant of order $s$ as the smallest number $\delta_s$ such that
\begin{equation}\label{eqn:RIP2del}
(1-\delta_s) \|x\|_2^2 \leq \|A x \|_2^2 \leq (1+\delta_s) \|x\|_2^2 
\mbox{ for all } x \in \Si_s,
\end{equation}
where $\Sigma_s$ is the set of all $s$-sparse vectors in $\C^n$ and $\|\cdot\|_2$ denotes the usual $\ell_2$-norm. 
If $\delta_s$ is sufficiently small we say that $A$ satisfies the RIP.
If $\delta_s < \delta_0$ for some suitably small $\delta_0$, then given measurements $y = A \hat{x} + e$ with $\|e\|_2 \leq \varepsilon$, the $\ell_2$-constrained $\ell_1$-minimization
program (also known as basis pursuit denoising) 
\[
\min_{z\in \C^n} \|z\|_1 \quad \mbox{ subject to } \quad  \|Az - y\|_2 \leq \varepsilon
\]
recovers a vector $x^\sharp$ which satisfies
\begin{equation}\label{error:bound2}
\|\hat{x} - x^\sharp\|_2 \lesssim s^{-1/2}\si_s(\hat{x})_1 + \frac{\eps}{m^{1/2}},
\end{equation}
where $\si_s(\hat{x})_1 = \inf_{\|z\|_0 \leq s} \|\hat{x}-z\|_1$ is the error of best $s$-term approximation to $\hat{x}$ in $\ell_1$.
A (scaled) Gaussian random matrix satisfies the RIP with high probability provided that
\begin{equation}\label{bound:RIP2}
m \geq C s \log(en/s),
\end{equation}
where $C>0$ is an absolute constant. This bound is optimal, see also below. 

In certain cases it is of interest to measure the level of noise in $\ell_p$-norms with $p$ different from $2$ and to study the corresponding $\ell_p$-constrained
basis pursuit denoising program
\begin{equation}\label{BPDNp}
\min_{z\in \C^n} \|z\|_1 \quad \mbox{ subject to } \quad
\|y-Az\|_p\leq \eps \tag{$\operatorname{BPDN}_p$}.
\end{equation}
The case $p=\infty$ appears, for instance, in quantized compressed sensing \cite{JHF11}, where $\ell_\infty$-constrained basis pursuit can ensure consistent reconstruction, see also Section~\ref{sec:quantizedCS} below. The program for $p=1$ is more robust to outliers than standard basis pursuit denoising. Also, when considering random measurement noise, different values of $p$ are appropriate depending on the distribution of the noise (see e.g.\ \cite{JHF11},\cite{Fuc09}). For example, $p=1$ is well-suited for double-exponential noise, whereas $p=2$ is appropriate for Gaussian noise.

Previous attempts in analyzing (BPDN$_p$) have used RIP conditions of the form
\begin{equation}\label{RIPpq}
c\|x\|_q \leq \|Ax\|_p\leq C\|x\|_q, \quad \operatorname{for \ all} \ x\in \Si_s. \tag{$\operatorname{RIP}_{p,q}$}
\end{equation}
It is part of the folklore in compressive sensing that  (RIP$_{p,q}$)
implies stable and robust recovery via (BPDN$_p$), with an
$\ell_q$-bound on the reconstruction error (see \cite{BGI08,JHF11} for special cases). 
Unfortunately, all available results on the number of required measurements for Gaussian and other random matrices ensuring (RIP$_{p,q}$)
scale significantly worse than \eqref{bound:RIP2} when $p \neq 1,2$. For certain values of $p$ and $q$, 
there are even negative results available which state that no matrix whatsoever can satisfy (RIP$_{p,q}$) in the optimal parameter regime
\eqref{bound:RIP2}. A more detailed overview is given below.

The purpose of this paper is to illuminate the discrepancy between on the one hand, the requirements needed for a matrix to satisfy an RIP condition 
of the form (RIP$_{p,q}$) and on the other hand, the conditions under which one can stably and
robustly recover any $s$-sparse (or approximately $s$-sparse) vector
$\hat{x}\in \C^n$ from noisy linear measurements $y=A\hat{x}+e$ via the generalized basis pursuit denoising program  
(BPDN$_p$). Our results show that a study of the statistical properties of
(BPDN$_p$) via the $\ell_q$-robust null space property yields better
results than via (RIP$_{p,q}$), both in terms of the required number of measurements as well as the allowed distribution of the
random measurements. In particular, one can achieve stable and robust reconstruction with Gaussian random matrices
in the optimal parameter regime \eqref{bound:RIP2} for any $1\leq p\leq \infty$. This result extends to various random matrices with heavier-tailed 
entries such as exponential matrices, see Section~\ref{sec:generalDist} for more information. Our proof relies on the small ball method developed in \cite{KoM13,Sh14_1,Sh14_2,Sh14_3}.   

{\bf Notation.} The usual $\ell_p$-norm on $\C^n$ is denoted by $\|x\|_p = (\sum_{j=1}^n |x_j|^p)^{1/p}$ for $1 \leq p < \infty$ and $\|x\|_\infty = \max_{j=1,\hdots,n} |x_j|$. We let $B_{\ell_p^n}$ and $S_{\ell_p^n}$ denote the associated unit ball and unit sphere, respectively. The expression $\|x\|_0 := \#\{j : x_j \neq 0\}$ counts the number of nonzero coefficients of $x$.
The expectation of a random variable $Z$ is written $\E Z$ and the probability of an event $E$ is denoted by $\mathbb{P}(E)$.
The $L_p$-norm of a measurable function $f$ with respect to a measure $\mu$ is denoted by $\|f\|_{L_p(\mu)}$. A Rademacher random variable $\eps$ satisfies $\bP(\eps=1)=\bP(\eps=-1)=1/2$ and a Rademacher sequence is a sequence of independent Rademacher random variables. For $t\in \R$, $\lfloor t\rfloor$ is the largest integer smaller than $t$ and $\lceil t\rceil$ is the smallest integer larger than $t$. Finally, we write $A\lesssim B$ if $A\leq cB$ for a universal constant $c>0$.

\section{The relation between (RIP$_{p,q}$) and (BPDN$_p$)}

Let us first summarize the known results on (RIP$_{p,q}$) and its
implication for sparse recovery via (BPDN$_p$). As is well known and already described above, the
(RIP$_{2,2}$) property was introduced in compressed sensing by
Cand\`{e}s and Tao in \cite{cata05,CaT06}. They showed that an $m\ti n$ matrix $A$ scaled by $m^{-1/2}$ with i.i.d.\ standard Gaussian entries satisfies 
\eqref{eqn:RIP2del} with probability $1-\eta$ if $m\gtrsim \del_s^{-2}(s\log(en/s)+\log(\eta^{-1}))$. 
If $A$ has this property with $\del_s$ smaller than a fixed threshold and $\|y-A\hat{x}\|_2\leq\eps$, 
then any minimizer $x^{\#}$ for (BPDN$_{2}$) satisfies an $\ell_q/\ell_1$-guarantee of the form
\[
\|\hat{x}-x^{\#}\|_q \lesssim s^{1/q-1}\si_s(\hat{x})_1 + s^{1/q-1/2} m^{-1/2} \eps 
\]
for any $1\leq q \leq 2$; the special case $q=2$ is stated in \eqref{error:bound2}. In particular, if $\hat{x}$ is exactly $s$-sparse (so $\si_s(\hat{x})_1=0$) and $\eps=0$, then  $\hat{x}$ can be
reconstructed exactly. Conversely, it is known that $m\gtrsim
s\log(n/s)$ measurements are also necessary for exact reconstruction
of all $s$-sparse vectors (see e.g.\ \cite[Theorem 10.11]{FoR13}). \par
A very similar connection exists between (RIP$_{1,1}$) and (BPDN$_1$) \cite{BGI08}. Indeed, the adjacency matrix of a random left $d$-regular bipartite graph with $n$ left vertices and $m$ right vertices with probability $1-\eta$ satisfies an (RIP$_{1,1}$) condition of the form
$$(1-\del)^{1/2}\|x\|_1\leq d^{-1}\|Ax\|_1\leq \|x\|_1 \qquad \operatorname{for \ all} \ x\in \Si_s,$$
provided that $d=\lceil \del^{-1}\log(en/(s\eta))\rceil$ and $m\geq
c_{\del}s\log(en/(s\eta))$. As a consequence \cite[Theorem
12]{BGI08}, if $\|y-A\hat{x}\|_1\leq\eps$ then any minimizer $x^{\#}$ of (BPDN$_1$) satisfies the $\ell_1/\ell_1$ guarantee
$$\|\hat{x}-x^{\#}\|_1\leq C(\del)\Big(\si_s(\hat{x})_1 + \frac{\eps}{d}\Big),$$
where $C(\del)=O((1-2\del)^{-1})$ for $\del\uparrow 1/2$.\par 
Interestingly, the rescaled adjacency matrix $d^{-1}A$ does not
satisfy the (RIP$_{2,2}$) condition. In fact, any (RIP$_{2,2}$)-matrix
with binary entries must satisfy $m\geq s^2\log(en/s)$
\cite[Theorem~4.6.1]{Cha10}. Conversely, if $A$ is standard Gaussian, $m^{-1/2}A$ cannot satisfy an (RIP$_{1,1}$) condition for $m\sim s\log(en/s)$ \cite{BGI08}. To see this, one can consider $x=e_1$, $\tilde{x}=s^{-1}\sum_{i=1}^s e_i$, where the $e_i$ denote the standard basis vectors. 
Then $\|x\|_1=\|\tilde{x}\|_1=1$, but $\|Ax\|_1\sim \sqrt{s}\|A\tilde{x}\|_1$.\par 
The two positive results for $p=q=2$ and $p=q=1$ have triggered further research on (BPDN$_p$) via restricted isometry properties. In \cite{JHF11} it was shown that a standard $m\ti n$ Gaussian matrix with
\begin{equation}
\label{eqn:JHFCond}
m\gtrsim \Big(\del^{-2}s\log(en/(s\del)) + \del^{-2}\log(\eta^{-1})\Big)^{p/2} + (p-1)2^{p-1}
\end{equation}
satisfies an (RIP$_{p,2}$) property for $2\leq p<\infty$ of the form
$$(1-\del)^{1/2}\|x\|_2\leq \mu_p^{-1}\|Ax\|_p\leq (1+\del)^{1/2}\|x\|_2 \; \mbox{ for  all }  x\in \Si_s,$$
where $\mu_p = \E\|G\|_p$ and $G$ is a standard $m$-dimensional Gaussian random vector. If $A$ satisfies this property for sparsity levels $s,2s,3s$ with constants $\del_s,\del_{2s},\del_{3s}$ small enough (see \cite[Theorem 1]{BJK14} for a precise statement), then for all $\hat{x}\in \C^n$ with $\|y-A\hat{x}\|_p\leq\eps$, any minimizer $x^{\sharp}$ of (BPDN$_p$) satisfies an $\ell_2/\ell_1$-guarantee
$$\|\hat{x}-x^{\sharp}\|_2 \lesssim s^{-1/2}\si_s(\hat{x})_1 + \frac{\eps}{\mu_p}.$$
In \cite{AGR15} it is shown that the $m\ti n$ adjacency matrix $A$ of a random left $d$-regular bipartite graph with $n$ left vertices and $m$ right vertices with high probability satisfies an (RIP$_{p,p}$) property
\begin{equation*}
(1-\del)\|x\|_p^p \leq d^{-1}\|Ax\|_p^p\leq (1+\del)\|x\|_p^p \qquad \operatorname{for \ all} \ x\in \Si_s,
\end{equation*}
provided that, in the case $1<p<2$,
\begin{align*}
m & \geq C_p(s^p\del^{-2}+s^{4-2/p-p}\del^{-2/(p-1)})\log n,\\
d & \geq \tilde{C}_p(\del^{-1}s^{p-1} + s^{(p-1)/p}\del^{-1/(p-1)})\log n,
\end{align*}
where $C_p,\tilde{C}_p$ are singular for $p\downarrow 1$ and $p\uparrow 2$, or in the case $2<p<\infty$,
\begin{align*}
m & \geq p^{Cp} \del^{-2}s^p\log^{p-1}(n),\\
d & \geq p^{Cp}\del^{-1} s^{p-1}\log^{p-1}(n).
\end{align*}
If $A$ satisfies an RIP$_{p,p}$-property for $p>1$, then one can recover all $\hat{x}\in \C^n$ with $\|y-A\hat{x}\|_p\leq\eps$ via (BPDN$_p$) with an $\ell_p/\ell_1$-guarantee of the form
$$\|\hat{x}-x^{\sharp}\|_p \lesssim s^{1/p-1}\si_s(\hat{x})_1 + \eps,$$
see \cite[Theorem A.6]{AGR15} for a more precise statement. Interestingly, \cite{AGR15} also proved a lower bound on $m$ assuming that the 
$m \times n$ matrix satisfies (RIP$_{p,p}$). Their result \cite[Theorem 4.1]{AGR15} essentially shows that one needs at least $m\gtrsim s^p$ measurements for $p\neq 2$, so that the case $p=2$ should be considered a singularity. A straightforward modification of their argument shows that to satisfy (RIP$_{p,2}$) one needs at least $m\gtrsim s^{p/2}$, so that also the result in \cite{JHF11} (cf.\ (\ref{eqn:JHFCond})) cannot be improved significantly. We leave the verification of this implication to the interested reader.\par  
To summarize, two important phenomena occur when moving away from the familiar (RIP$_{2,2}$). First, one may need to consider different random matrix constructions to satisfy an RIP property with the optimal number of measurements. Second, the optimal scaling of the number of measurements in terms of the signal sparsity may dramatically worsen, especially for $p>2$. 

\section{Sparse recovery via BPDN$_{p}$: improved results}

One might think that the two phenomena concerning the (RIP$_{p,q}$) properties for $p \neq 2$ mentioned above, may carry over to recovery results via (BPDN$_p$) (see e.g.\ \cite{AGR15,JHF11}), in particular, that the minimal required number of measurements depends significantly worse than linear on the sparsity.
We will now show that rather the contrary is true: the scaling in terms of the sparsity generally does not worsen if $p\neq 2$ and, moreover, the optimal recovery results are realized by a standard Gaussian matrix.

Let us note that earlier work already identified a looseness in the
relation between the classical (RIP$_{2,2}$) and (BPDN$_2$). For
example, if $A$ has independent, isotropic, log-concave rows, then
(\ref{eqn:RIP2del}) is satisfied with high probability if $m \geq
c(\del) s\log^2(en/s)$ (cf. \cite{ALL11}), and the square in the $\log$-factor cannot be
removed (cf. \cite[Proposition~5.5]{MR2796091}). On the other hand, this matrix still satisfies,
with high probability, the exact reconstruction property for
$s$-sparse vectors via $\ell_1$-minimization in the optimal
measurement regime $m\simeq s\log (en/s)$
(cf. \cite[Theorem~7.3]{Kol11} -- see also \cite{MR3134335} for the
special case of measurement matrices with i.i.d.\
Weibull entries). More recently, near-matching necessary and sufficient conditions on the moments of the i.i.d. entries of a matrix to satisfy the exact
reconstruction property (and more generally, stable and robust
recovery via (BPDN$_2$)) in this regime were recently derived by the
second-named author and Mendelson \cite{LeM14}. We recover as a special case a variation of this (sufficient) result, see Corollary~\ref{cor:condEntries} below. Such a result cannot be
proved via an RIP-based analysis since the right-hand side of
(RIP$_{2,2}$), i.e.,
$$\|Ax\|_2\leq C\|x\|_2 \quad \mbox{for all } x\in \Si_s$$ 
requires either strong concentration properties or a larger number of measurements $m$ than the optimal number $s\log(en/s)$
(see the discussion in \cite{LeM14} and Section~\ref{sec:RIPRIP} for more details).\par   
For our analysis we let $X_1,\ldots,X_m$ be i.i.d.\ copies of a random
vector $X$ in $\C^n$, which is defined on a probability space $(\Om,\cA,\bP)$. Let $P_m$ be the associated empirical measure
$$P_m = \frac{1}{m} \sum_{i=1}^m \del_{X_i}.$$
The following observation follows immediately from the proof of
Theorem 2.1 in \cite{KoM13}, by replacing the ``Chebyshev'' bound
$$\|f\|^2_{L_2(P_m)} \geq u^2 P_m(|f|\geq u)$$
by 
$$\|f\|^p_{L_p(P_m)} \geq u^p P_m(|f|\geq u).$$\begin{lemma}
\label{lem:KoM13}
Fix $1\leq p<\infty$. Let $\cF$ be a class of functions from $\C^n$ into $\C$. Consider 
$$Q_{\cF}(u) = \inf_{f\in \cF} \bP(|f(X)|\geq u)$$
and 
$$R_m(\cF) = \E\sup_{f\in \cF}\Big|\frac{1}{m} \sum_{i=1}^m \eps_i f(X_i)\Big|,$$
where $(\eps_i)_{i\geq 1}$ is a Rademacher sequence. Let $u>0$ and $t>0$, then, with probability at least $1-2e^{-2t^2}$,
$$\inf_{f\in \cF} \frac{1}{m}\sum_{i=1}^m |f(X_i)|^p \geq u^p\Big(Q_{\cF}(2u) - \frac{4}{u}R_m(\cF) - \frac{t}{\sqrt{m}}\Big).$$
\end{lemma}
Consider the following sparse recovery problem: we take $m$ noisy
linear measurements of an (approximately) $s$-sparse signal $\hat{x}$, i.e., we
observe $y=A\hat{x}+e$ where $A\in \C^{m\ti N}$ and we suppose that
the noise satisfies $\|e\|_p\leq \eps$. We aim to recover $\hat{x}$
from $y$ via (BPDN$_{p}$). For the analysis we recall the following
standard notion (cf. for instance \cite[Definition~4.21]{FoR13}). Given $q\geq 1$, we say that $A$ satisfies the $\ell_q$-robust null space property of order $s$ with constants $0<\rho<1$ and $\tau>0$ with respect to a norm $\|\cdot\|$ if for any set $S\subset[n]$ with $|S|\leq s$ and any $x\in \C^n$, 
$$\|x_S\|_q\leq \frac{\rho}{s^{1-1/q}}\|x_{S^c}\|_1 + \tau \|Ax\|.$$
If $A$ has this property, then any solution $x^{\#}$ to
$$\min_{z\in \C^n}\|z\|_1 \quad \mbox{ subject  to } \quad \|y-Az\|\leq \eps$$
satisfies, for any $1\leq r\leq q$, the reconstruction error bound
$$\|\hat{x}-x^{\#}\|_r\leq C_{\rho} s^{1/r-1}\si_s(\hat{x})_1 + \tau D_{\rho} s^{1/r-1/q}\eps,$$
with $C_{\rho}=(1+\rho)^2/(1-\rho)$ and $D_{\rho} =
(3+\rho)/(1-\rho)$ when $\norm{e}\leq \eps$ (cf. \cite[Theorem~4.25]{FoR13}).\par
To analyze $\ell_q$-robust null space properties, we introduce the cone
\begin{equation*}
T_{\rho,s}^q = \big\{x\in\C^n: \exists S \subset[n],\; |S|=s : \|x_S\|_q\geq\frac{\rho}{s^{1-1/q}}\|x_{S^c}\|_1\big\}. 
\end{equation*}
Note that $T_{\rho,s}^q$ contains $\Si_s$. We use the following observation.
\begin{lemma}\label{lem:Dsq}
Fix $1\leq q<\infty$. Set
$$\Si_s^q:=\{x\in \C^n \ : \ \|x\|_0\leq s, \ \|x\|_q=1\}$$
and let $D_s^q$ be its convex hull. Then $D_s^q$ is the unit ball with respect to the norm
$$\|x\|_{D_s^q}:=\sum_{\ell=1}^{\lceil n/s \rceil}\Big(\sum_{i\in I_{\ell}} x_i^{*q}\Big)^{1/q},$$
where $I_1,\ldots,I_{\lceil n/s \rceil}$ form a uniform partition of $[n]$, i.e., 
$$I_{\ell}=\left\{\begin{array}{ll}
\{s(\ell-1)+1,\ldots,s\ell\}, & \ell=1,\ldots, \lceil n/s \rceil-1,\\
\{s(\lceil n/s \rceil-1)+1,\ldots, n\}, & \ell=\lceil n/s \rceil,
\end{array}\right.
$$
and $x^*$ is the nonincreasing rearrangement of $x$. As a consequence,
\begin{equation*}
T_{\rho,s}^q\cap B_{\ell_q^n}\subset (2+\rho^{-1})D_s^q.
\end{equation*}
\end{lemma}
\begin{proof}
We  proceed by making straightforward modifications to the proof of
\cite[Lemma 3]{KaR15} (see also \cite[Lemma
4.5]{RuV08} or \cite{MR2373017}), which corresponds to the case $q=2$.  

A vector $x \in D_s^q$ can be represented as $x=\sum_i \alpha_i  x_i$ with $\alpha_i\geq0$, $\sum_i\alpha_i=1$ and $x_i\in S_{\ell_q^n}$, $\|x_i\|_0\leq s$. In particular, $\|x_i\|_{D_s^q}=\|x_i\|_q=1$. By the triangle inequality
$$
\|x\|_{D_s^q}\leq\sum_i\alpha_i\|x_i\|_{D_s^q}=\sum_i\alpha_i=1,
$$
so $D_s^q$ is contained in the $\|\cdot\|_{D_s^q}$-unit ball. To prove the
reverse inclusion, suppose that $\|x\|_{D_s^q}\leq 1$. We partition
the index set $[n]$ into subsets $S_1$, $S_2$, \ldots of size $s$,
such that $S_1$ corresponds to the indices of the $s$ largest entries of $x$, $S_2$ to the next $s$ ones, etc. Set $\alpha_i=\|x_{S_i}\|_q$. Then $x$ can be written as
$$
x=\sum_{i:\alpha_i\neq 0}\alpha_i(\alpha_i^{-1} x_{S_i}),
$$
where
$$
\sum_{i:\alpha_i\neq 0}\alpha_i=\sum_i \|x_{S_i}\|_q=\|x\|_{D_s^q}\leq 1.
$$
Clearly, for any $\alpha_i\neq 0$, $\|\alpha_i^{-1}x_{S_i}\|_q=1$ and
$\|\alpha_i^{-1}x_{S_i}\|_0\leq s$, so $x\in D_s^q$.

To prove the second statement, fix $x \in T_{\rho,s}^q\cap B_{\ell_q^n}$ and write
\begin{equation}
\|x\|_{D_s^q} = \Big(\sum_{i\in I_1}
x_i^{*q}\Big)^{1/q}+\Big(\sum_{i\in I_2} x_i^{*q}\Big)^{1/q}+\sum_{\ell\geq 3} \Big(\sum_{i\in I_{\ell}} x_i^{*q}\Big)^{1/q}\label{eqn:splitDqs}.
\end{equation}
To bound the last term, note that for each $i\in I_{\ell}$, $\ell\geq 3$,
$$
x^*_i\leq\frac{1}{s}\sum_{j\in I_{\ell-1}}x^*_j\quad\text{and}\quad \Big(\sum_{i\in I_{\ell}} x_i^{*q}\Big)^{1/q}\leq\frac{1}{s^{1-1/q}}\sum_{j\in I_{\ell-1}}x^*_j.
$$
Summing up over $\ell\geq 3$ yields
$$
\sum_{\ell\geq 3} \Big(\sum_{i\in I_\ell} x_i^{*q}\Big)^{1/q}\leq\frac{1}{s^{1-1/q}}\sum_{\ell\geq 2}\sum_{j\in I_\ell}x^*_j.
$$
Since $x \in T_{\rho,s}^q\cap B_{\ell_q^n}$, there is an $S\subset [n]$ with $|S|=s$, such that $\|x_S\|_q\geq\frac{\rho}{s^{1-1/q}}\|x_{S^c}\|_1$. 
Therefore,
\begin{equation*}
\sum_{\ell\geq 2}\sum_{i\in I_\ell}x^*_i \leq\|x_{S^c}\|_1\leq \frac{s^{1-1/q}}{\rho}\|x_S\|_q \leq\frac{s^{1-1/q}}{\rho}\Big(\sum_{i\in
I_1} x_i^{*q}\Big)^{1/q},
\end{equation*}
where we used that in the worst case $S$ corresponds to $s$ largest absolute coefficients of $x$. It follows that
$$
\sum_{\ell\geq 3} \Big(\sum_{i\in I_\ell}
x_i^{*q}\Big)^{1/q}\leq\frac{1}{\rho}\Big(\sum_{i\in I_1} x_i^{*q}\Big)^{1/q}.
$$
Since $\|x\|_q\leq 1$, (\ref{eqn:splitDqs}) implies that $\|x\|_{D_s^q}\leq 2+\rho^{-1}$.
\end{proof}

We are now prepared to prove the main result of this article. To keep our exposition accessible, we first consider the special case of a standard Gaussian random matrix, i.e., a matrix with independent
normally distributed entries with mean zero and variance one. In Section~\ref{sec:generalDist} we generalize our result to a wider class of random matrices.
\begin{theorem}
\label{thm:main}
Let $A$ be an $m\ti n$ standard Gaussian matrix. Fix $1\leq p\leq \infty$, $q\geq 2$ and $0<\eta<1$. Suppose that
\begin{equation}\label{main:scale:m}
m \gtrsim s^{2-2/q}\log(en/s) + \log(\eta^{-1}).
\end{equation}
Then, with probability exceeding $1-\eta$ the following holds: for any $\hat{x}\in \C^n$ and $y=A\hat{x}+e$, where $\|e\|_p\leq \eps$, any solution $x^{\#}$ to (BPDN$_{p}$) satisfies
$$\|\hat{x}-x^{\#}\|_r\lesssim s^{1/r-1}\si_s(\hat{x})_1 + s^{1/r-1/q}\frac{\eps}{m^{1/p}},$$
for any $1\leq r\leq q$.  
\end{theorem}
\begin{remark} The most interesting case in the above theorem is $q=2$. Then the optimal scaling $m \geq C s \log(en/s)$ implies that with high probability we obtain
the error bound
$$\|\hat{x}-x^{\#}\|_2\lesssim s^{-1/2}\si_s(\hat{x})_1 + m^{-1/p} \eps$$
for reconstruction via $\ell_p$-constrained basis pursuit.

For $q>2$ the scaling \eqref{main:scale:m} of $m$
in terms of the sparsity is near-optimal. Indeed, it is known \cite[p.~213]{pi85} that for
$q>2$ and $m\leq n-1$, the Gelfand width of $B_{\ell_1^n}$ in $\ell_q^n$
satisfies
$$d^m(B_{\ell_1^n},\ell_q^n)\geq d^m(B_{\ell_1^n},\ell_{\infty}^n)\geq c m^{-1/2}.$$
Thus, if we want to satisfy $\|\hat{x}-x^{\#}\|_q\lesssim
s^{1/q-1}\si_s(\hat{x})_1$ for all $\hat x\in\C^n$, then it is
necessary (cf. \cite[Theorem~10.4]{FoR13}) that $m^{-1/2}\lesssim
s^{1/q-1}$ or $m\gtrsim s^{2-2/q}$. Thus, up to possibly a
logarithmic factor we cannot improve the scaling of $m$ in terms of the
sparsity in Theorem~\ref{thm:main}.
\end{remark}
\begin{proof}[Proof of Theorem~\ref{thm:main}]
Suppose first that $p<\infty$. As was noted before, it suffices to
show that with probability at least $1-\eta$ the $\ell_q$-robust null
space property of order $s$ holds with respect to $\ell_p^n$-norm, with parameters $\rho$ and $\tau/m^{1/p}$ for some $0<\rho<1$ and $\tau>0$. Let us first observe that it suffices to show that 
\begin{equation}
\label{eqn:claim1}
\bP\Big(\inf_{x\in T_{\rho,s}^q\cap S_{\ell_q^n}} \|Ax\|_p \geq \frac{m^{1/p}}{\tau}\Big)\geq 1-\eta.
\end{equation}
Indeed, if this is true, then with probability at least $1-\eta$ the following holds: if $x\in \C^n$ satisfies $\|Ax\|_p< (m^{1/p}/\tau)\|x\|_q$ then $x/\|x\|_q$ is not in $T_{\rho,s}^q$. Therefore, for any $S\subset[n]$ with $|S|\leq s$, 
$$\|x_S\|_q \leq \frac{\rho}{s^{1-1/q}}\|x_{S^c}\|_1 \leq \frac{\rho}{s^{1-1/q}}\|x_{S^c}\|_1 + \frac{\tau}{m^{1/p}}\|Ax\|_p.$$
On the other hand, if $x\in \C^{n}$ satisfies $\|Ax\|_p\geq (m^{1/p}/\tau)\|x\|_q$, then trivially
$$\|x_S\|_q \leq \|x\|_q \leq \frac{\rho}{s^{1-1/q}}\|x_{S^c}\|_1 + \frac{\tau}{m^{1/p}}\|Ax\|_p.$$
To prove (\ref{eqn:claim1}), we write 
$$\inf_{x\in T_{\rho,s}^q\cap S_{\ell_q^n}} \frac{\|Ax\|_p}{m^{1/p}} = \inf_{x\in T_{\rho,s}^q\cap S^{n-1}_q} \Big(\frac{1}{m}\sum_{i=1}^m |\langle X_i,x\rangle|^p\Big)^{1/p},$$
where $X_i$ denotes the $i$-th row of $A$. To apply Lemma~\ref{lem:KoM13}, we estimate the small ball probability $Q_{\cF}$ and the expected Rademacher supremum $R_m(\cF)$ for the set of linear functions
$$\cF = \{\langle \cdot,x\rangle \ : \ x\in T_{\rho,s}^q\cap S_{\ell_q^n}\}.$$
Let $V=m^{-1/2}\sum_{i=1}^m \eps_i X_i$, then by Lemma~\ref{lem:Dsq}, 
\begin{align*}
R_m(\cF) & = m^{-1/2} \E\sup_{x\in T_{\rho,s}^q\cap S_{\ell_q^n}} \langle V,x\rangle \\
& \leq (2+\rho^{-1})m^{-1/2} \E\sup_{x\in D_s^q}\langle V,x\rangle\\
& = (2+\rho^{-1})m^{-1/2} \E\sup_{x\in \Si_s^q}\langle V,x\rangle,
\end{align*}
as $D_s^q$ is the convex hull of $\Si_s^q$. Since any $x\in \Si_s^q$ satisfies $\|x\|_2\leq s^{1/2-1/q}\|x\|_q$,  
$$R_m(\cF) \leq s^{1/2-1/q}(2+\rho^{-1})m^{-1/2} \E\sup_{x\in \Si_s^2}\langle V,x\rangle.$$
Since $X_1,\ldots,X_m$ are independent standard Gaussian vectors, so is $V$. Thus,
$$\E\sup_{x\in \Si_s^2}\langle V,x\rangle = w(\Si_s^2),$$
the Gaussian width of $\Si_s^2$. It is known that
$$w(\Si_s^2) \leq \sqrt{2s \log(en/s)} + \sqrt{s},$$ 
see e.g.~\cite[Lemma~4]{KaR15}, and we can conclude that
$$R_m(\cF) \leq c s^{1-1/q}(2+\rho^{-1})m^{-1/2}\sqrt{\log(en/s)}.$$

To estimate the small ball probability, note that, since $\|x\|_q \leq \|x\|_2$, for any $x\in S_{\ell_q^n}$,
\begin{align*}
\bP(|\langle X_i,x\rangle|\geq u) & = \bP\Big(\Big|\Big\langle
X_i,\frac{x}{\|x\|_2}\Big\rangle\Big|\geq \frac{u}{\|x\|_2}\Big)\\
 &\geq \bP\Big(\Big|\Big\langle X_i,\frac{x}{\|x\|_2}\Big\rangle\Big|\geq
u\Big) = \bP(|g|\geq u),
\end{align*}
where $g$ is a standard Gaussian real-valued random variable. Therefore,
$$Q_{\cF}(2u)\geq \bP(|g|\geq 2u).$$
Now pick $u_*$ small enough so that the right hand side is bigger than $1/2$, say. Pick $m$ large enough so that
$$\max\Big\{\frac{4c(2+\rho^{-1})s^{1-1/q}\sqrt{\log(en/s)}}{u_*\sqrt{m}}, \frac{\sqrt{\log(2/\eta)}}{\sqrt{2m}}\Big\}\leq 1/8.$$
By Lemma~\ref{lem:KoM13} we can now conclude that (\ref{eqn:claim1}) holds with $\tau=4^{1/p}/u_*$.

Finally, let $p=\infty$. Since $\|Ax\|_{\log m}\leq e\|Ax\|_{\infty}$,
\begin{equation*}
\bP\Big(\inf_{x\in T_{\rho,s}^q\cap S_{\ell_q^n}} \|Ax\|_{\infty}  \geq \frac{1}{\tau}\Big) \geq \bP\Big(\inf_{x\in T_{\rho,s}^q\cap S_{\ell_q^n}} \|Ax\|_{\log m} \geq \frac{e}{\tau}\Big).
\end{equation*}
Thus, in this case the result follows from our proof for $p=\log m$.     
\end{proof}

 \section{Application to quantized compressed sensing}
 \label{sec:quantizedCS}

Consider the situation where we quantize noiseless compressed sensing measurements using a uniform scalar quantization scheme. That is, we observe $y=Q_{\theta}(A \hat{x})$, where $Q_{\theta}:\R^m\rightarrow (\theta\bZ + \theta/2)^m$ is the uniform quantizer with bin width $\theta$ defined by $Q_{\theta}(z) = \big(\theta\lfloor z_i/\theta\rfloor + \theta/2\big)_{i=1}^m$. Graphically, we divide $\R^m$ into hypercubes (or `bins') with side length $\theta$ and map $A\hat{x}$ to the center of the hypercube in which it resides. We view the quantized measurements as noisy linear measurements 
$y=A \hat{x} + e$, by setting $e=Q_{\theta}(A \hat{x}) - A \hat{x}$. 
Since the bin width of the quantization is $\theta$, we clearly have $\|e\|_{\infty}\leq \theta/2$. 

To obtain a satisfactory reconstruction $x^{\#}$ of the signal, we
would like to ensure that it is \emph{quantization consistent}. This
means that we require that $y=Q_{\theta}(A x^{\#})$. If we define 
$$B_{\theta}=\{z\in\R^m \ : \ -\theta/2\leq z_i<\theta/2, \ i=1,\ldots,m\},$$ 
then $x^{\#}$ is quantization consistent if and only if $A x^{\#}-y \in B_{\theta}$. Thus, we should solve the following \emph{quantization consistent basis pursuit} program
\begin{equation}\label{QCBP}
\min_{z\in \R^n} \|z\|_1 \quad \mbox{ subject to } \quad Az-y \in B_{\theta} \tag{$\operatorname{QCBP}$}.
\end{equation}
This program is strongly related to (BPDN$_{\infty}$) with $\eps=\theta/2$ (which correspond to taking the closure $\overline{B_{\theta}}$ instead of $B_{\theta}$ in (QCBP)). In fact, either 1) a minimizer for (QCBP) exists, this is then also a minimizer for (BPDN$_{\infty}$), or 2) no minimizer exists, in which case every minimizer of (BPDN$_{\infty}$) is quantization inconsistent. In particular, Theorem~\ref{thm:main} implies the following statement.
\begin{corollary}
\label{cor:quantization}
Let $A$ be an $m\ti n$ standard Gaussian matrix and $0<\eta<1$. Suppose that
$$m\gtrsim s\log(en/s) + \log(\eta^{-1}).$$
Then, with probability exceeding $1-\eta$ the following holds: for any $\hat{x}\in \R^n$ and quantized measurements $y=Q_{\theta}(A\hat{x})$, any solution $x^{\#}$ to (QCBP) is a quantization consistent reconstruction of $\hat{x}$ and satisfies the error bound
$$\|\hat{x}-x^{\#}\|_2 \lesssim s^{-1/2}\si_s(\hat{x})_1 + \theta.$$ 
\end{corollary}
Comparing Corollary~\ref{cor:quantization} to the performance of the usual basis pursuit denoising, (BPDN$_{2}$), we can still reconstruct with the optimal number of measurements, but the reconstruction error does not decay beyond (a constant multiple of) the quantization precision $\theta$.

Let us compare to the work in \cite{JHF11}, where the authors introduced and analyzed (BPDN$_{p}$) with $2\leq p<\infty$ for the purpose of recovering a signal from quantized measurements (as described above). They did not obtain a result for $p=\infty$, but the idea is that the reconstruction becomes more consistent as $p\rightarrow\infty$. A main result in \cite{JHF11} shows the following, via an (RIP$_{p,2}$)-based analysis. Assume that the error vector $e$ consists of i.i.d.\ $U([-\theta/2,\theta/2])$ random variables, that is we assume that the quantization error is uniformly distributed in each bin (this is called the \emph{high resolution assumption}). With probability at least $1-e^{-2t^2}$, 
$$\|e\|_p\leq \eps_p := \frac{\theta}{2(p+1)^{1/p}} (m+t(p+1)\sqrt{m})^{1/p}.$$
This suggests to try to recover $\hat{x}$ via (BPDN$_{p}$) with $\eps=\eps_p$.  Let $A$ be an $m\ti n$ standard Gaussian matrix with
\begin{equation}
\label{eqn:measRIPp2}
m\gtrsim (ps\log(en\sqrt{p}/s) + p\log(\eta^{-1}))^{p/2},
\end{equation} 
then with probability at least $1-\eta$, for any $\hat{x}\in\R^n$ the reconstruction $x^{\#}$ via (BPDN$_{p}$) with $y=Q_{\theta}(A\hat{x})$ and $\eps=\eps_p$ satisfies
$$\|\hat{x}-x^{\sharp}\|_2 \lesssim s^{-1/2}\si_s(\hat{x})_1 + \frac{\theta}{\sqrt{p+1}}.$$
Compared to Corollary~\ref{cor:quantization}, the reconstruction error
due to quantization error shows decay with $p$. Note, however, that
the value we can take for $p$ is implicitly limited by
(\ref{eqn:measRIPp2}), and in particular we cannot set $p=\infty$ so
that $x^{\sharp}$ is not guaranteed to be quantization
consistent. Moreover, when $p > 2$, the number of required measurements grows faster than linear in the sparsity. In fact, it grows exponentially in $p$, as opposed to the minimal number of
measurements needed in Corollary~\ref{cor:quantization}.
 
 \section{Generalization to different distributions}
 \label{sec:generalDist}
 
From the proof of Theorem~\ref{thm:main} we extract the following statement, which allows us to generalize our recovery result (as well as Corollary~\ref{cor:quantization}) to a variety of random matrices beyond the Gaussian case, while retaining the same (optimal) recovery guarantees as for a standard Gaussian matrix.
\begin{theorem}
\label{thm:genDistributions}
Let $A$ be an $m\ti n$ random matrix with i.i.d.\ rows $X_1,\ldots,X_m$ which are distributed as $X$. Suppose that for some $u_*>0$ and $\beta>0$,
\begin{equation}
\label{eqn:SBass}
\bP\big[|\langle X,x\rangle|\geq u_*\big]\geq \beta \qquad \operatorname{for \ all \ } x\in S_{\ell_2^n},
\end{equation}
and, if $V=m^{-1/2}\sum_{i}\varepsilon_i X_i$ then for some $\ka>0$,
\begin{equation*}
\E\sup_{x\in \Si_s^2}\langle V,x\rangle = \E\Big(\sum_{i=1}^s (V_i^*)^{2}\Big)^{1/2}\leq \ka \sqrt{s\log(en/s)},
\end{equation*}
where $V_1^*\geq\ldots\geq V_n^*$ is the nonincreasing rearrangement of $V$. Fix $1\leq p\leq \infty$ and $q\geq 2$. If 
$$m\gtrsim \max\Big\{\frac{\ka^2}{u_*^2\beta^2} s^{2-2/q}\log(en/s),\frac{\log(\eta^{-1})}{\beta^2}\Big\},$$
then with probability at least $1-\eta$ the following holds: for any $\hat{x}\in \C^n$ and $y=A\hat{x}+e$, where $\|e\|_p\leq \eps$, any solution $x^{\#}$ to (BPDN$_{p}$) satisfies
$$\|\hat{x}-x^{\#}\|_r\lesssim s^{1/r-1}\si_s(\hat{x})_1 + s^{1/r-1/q}\frac{\eps}{\beta^{1/p}u_*m^{1/p}},$$
for any $1\leq r\leq q$.  
\end{theorem}
To verify the small ball condition (\ref{eqn:SBass}), it is often useful to apply the Paley-Zygmund inequality
\begin{equation}
\label{eqn:P-Z}
\bP(\zeta>t) \geq \frac{(\E\zeta -t)^2}{\E\zeta^2}, \qquad 0\leq t\leq \E\zeta,
\end{equation}
which holds for any nonnegative random variable $\zeta$. In particular, if $X$ is a random vector with independent, mean-zero entries $\xi_1,\ldots,\xi_n$ which have variance $\si^2$ and fourth moment bounded by $\mu^4$, then
\begin{equation}
\label{eqn:P-Zindsum}
\bP(|\langle X,x\rangle|>t) \geq \frac{(\si^2-t^2)^2}{\mu^4}, \qquad 0\leq t\leq \si,
\end{equation}
whenever $\|x\|_2=1$. We refer to \cite[Lemmas 7.16 and 7.17]{FoR13} for details.\par
Let us now verify the conditions of Theorem~\ref{thm:genDistributions} for some concrete classes of matrices.
\begin{corollary}
\label{cor:subgaussian}
Suppose that the rows of $A$ are i.i.d.\ copies of $X$, where $X$ is 
\begin{itemize}
\item sub-isotropic, i.e., $\E \langle X,x\rangle^2 \geq \|x\|_2^2$ for all $x \in \C^n$;
\item 1-subgaussian, i.e., $\E\exp(t\langle X,x\rangle)\leq \exp(t^2)$ for all $x\in \C^n$ with $\|x\|_2\leq 1$ and $t\in \R$.
\end{itemize}
If $m\gtrsim s^{2-2/q}\log(en/s) + \log(\eta^{-1})$ then the conclusion of Theorem~\ref{thm:main} holds. 
\end{corollary}
\begin{proof}
We verify the two conditions of Theorem~\ref{thm:genDistributions}. To verify (\ref{eqn:SBass}), we use (\ref{eqn:P-Z}) for $|\langle X,x\rangle|^2$ to get
$$\bP(|\langle X,x\rangle|>u)\geq \frac{(\E|\langle X,x\rangle|^2 - u^2)^2}{\E|\langle X,x\rangle|^4} \geq (1-u^2)^2,$$
whenever $0\leq u\leq 1$ and $\|x\|_2=1$. In the last inequality, we used that $X$ is sub-isotropic and subgaussian.\par
To verify the second condition, note that by assumption, the random variable $\langle X_i,x-y\rangle$ is 2-subgaussian for any $x,y\in \Si_s^2$. Therefore $V=m^{-1/2}\sum_{i}\varepsilon_i X_i$ is a $4$-subgaussian random vector (see e.g.\ \cite[Theorem 7.27]{FoR13}). By Dudley's inequality (see e.g.\ \cite[Theorem 8.23]{FoR13}),
$$\E\sup_{x\in \Si_s^2}\langle V,x\rangle \lesssim \int_0^1 [\log(\cN(\Si_s^2,\|\cdot\|_2,u)]^{1/2} \ du.$$
Since for any $u>0$
\begin{equation*}
\cN(\Si_s^2,\|\cdot\|_2,u) \leq {n\choose s} \max_{S\subset [n]: \ |S|\leq s} \cN(B_S,\|\cdot\|_2,u) \leq (en/s)^s (1+(2/u))^s,
\end{equation*}
we conclude that
\begin{equation*}
 \E\sup_{x\in \Si_s^2}\langle V,x\rangle \lesssim \sqrt{s\log(en/s)} + \sqrt{s} \int_0^1 [\log(1+(2/u))]^{1/2} \ du \lesssim \sqrt{s\log(en/s)} .
\end{equation*}
\end{proof}
The following result concerns matrices with i.i.d.\ entries.
\begin{corollary}
\label{cor:condEntries}
Suppose that $X=(\xi_1,\ldots,\xi_n)$, with the $\xi_i$ independent, mean-zero and identically distributed as $\xi$. Suppose that for some $\la>0$ and $\al\geq 1/2$,
\begin{equation}
\label{eqn:momentLogn}
(\E|\xi|^r)^{1/r} \leq \la r^{\al}, \qquad \operatorname{for \ all \ } 2\leq r \leq \log n.
\end{equation}
and that (\ref{eqn:SBass}) holds. If
\begin{equation*}
m \gtrsim \max\Big\{\frac{\la^2 e^{4\al-2}}{u_*^2\beta^2}s^{2-2/q}\log(en/s),\frac{\log(\eta^{-1})}{\beta^2},(\log(n))^{2\al-1}\Big\},
\end{equation*}
then the conclusion of Theorem~\ref{thm:genDistributions} holds. 
\end{corollary}
Specializing Corollary~\ref{cor:condEntries} to $p=q=2$, we obtain a result similar to \cite[Theorem A]{LeM14}. Let us compare the two results. On the one hand, our result gives a better power in the $\log(n)$ factor ($2\al-1$ versus $4\al-1$) and improved (actually optimal) dependence on the failure probability $\eta$. On the other hand, \cite[Theorem A]{LeM14} does not require independence of the $\xi_i$ and needs only a small ball assumption on the set of sparse vectors $\Si_s$ (rather than one on the larger set $T_{\rho,s}^2\cap S_{\ell_2^n}$ used here).
\begin{proof}
We fix the randomness in the Rademacher sequence $(\eps_i)$. The random variables $V_j=m^{-1/2}\sum_{i=1}^m \eps_i X_{ij}$ are then independent and mean-zero. Since $X_{ij}$ satisfies (\ref{eqn:momentLogn}), \cite[Lemma 2.8]{LeM14} shows that if $m\geq (\log(n))^{\max\{2\al-1,1\}}$, then for any $2\leq p\leq \log(n)$
$$(\E|V_j|^p)^{1/p} \lesssim e^{2\al-1}\la\sqrt{p},$$
i.e., the first $\log(n)$ moments show subgaussian behaviour. Therefore, (the proof of) \cite[Lemma 6.5]{Sh14_1} shows that
$$\E\Big(\sum_{i=1}^s (V_i^*)^{2}\Big)^{1/2}\lesssim e^{2\al-1}\la\sqrt{s\log(en/s)}.$$
The result is now immediate from Theorem~\ref{thm:genDistributions}.  
\end{proof}
\begin{example}
\label{exa:iid}
Let $A$ be a random matrix with i.i.d. entries $A_{ij}$. Below we list some conditions under which the conclusion of Theorem~\ref{thm:main} is valid. Note that if we measure the reconstruction error in $\ell_2$ (i.e., $q=2$), then the stated lower bounds always coincide with the optimal number of measurements.
\begin{enumerate}
\renewcommand{\labelenumi}{(\roman{enumi})} 
\item If the $A_{ij}$ are random signs (i.e.\ Rademachers), then $m\gtrsim s^{2-2/q}\log(en/s) + \log(\eta^{-1})$ is sufficient for the recovery guarantee in Theorem~\ref{thm:main}. This follows from Corollary~\ref{cor:condEntries} with $\la=1$, $\al=1/2$ and $\beta,u_*$ universal constants.
\item If the $A_{ij}$ are standard symmetric exponential random variables, then $m\gtrsim s^{2-2/q}\log(en/s) + \log(\eta^{-1})$ suffices for the recovery guarantee in Theorem~\ref{thm:main}. Indeed, in this case one can apply Corollary~\ref{cor:condEntries} with $\la=\al=1$ and take for $\beta,u_*$ universal constants. 
\item Suppose that the $A_{ij}$ are distributed as a random variable $\xi$, which has probability density function 
$$p(x) = \frac{\ga-1}{2\ga} \min\{1,|x|^{-\ga}\}, \qquad x\in\R,$$
for some $\ga>1$. One readily calculates that 
$$\E|\xi|^p = \frac{\ga-1}{\ga}\Big(\frac{1}{\ga-p-1} + \frac{1}{p+1}\Big)$$
for $p<\ga-1$ and $\E|\xi|^p=\infty$ for $p\geq \ga-1$. If we assume $\ga\geq \log(n)+2$, say, then $\xi$ trivially satisfies the moment bound in Corollary~\ref{cor:condEntries} with $\al=1/2$. Moreover, if $\ga>5$ then $\E\xi^2 = (\ga-1)/(3\ga-9)$ and $\E\xi^4=(\ga-1)/(5\ga-25)$ so the Paley-Zygmund inequality (\ref{eqn:P-Zindsum}) implies that (\ref{eqn:SBass}) holds for universal constants $u_*,\beta$ if $\ga\geq 6$, say. In conclusion, if we assume
$$\ga\geq \max\{\log(n)+2,6\},$$
then $m\gtrsim s^{2-2/q}\log(en/s) + \log(\eta^{-1})$ is sufficient for the recovery guarantee in Theorem~\ref{thm:main}. 
\end{enumerate}
\end{example}
The last example illustrates that only the behaviour of the first $\log n$ moments of the entries of $A$ is important for our sparse recovery result, the higher moments need not even exist.\par
To conclude this section, we extend the example of a standard symmetric exponential matrix (part (ii) of Example~\ref{exa:iid}) to matrices with i.i.d.\ isotropic, unconditional, log-concave rows. In particular, we do not assume that the entries within a row are independent. Recall that a probability measure $\mu$ on $\R^n$ is called log-concave if for any (Borel) sets $A,B\subset\R^n$ and $0\leq \theta\leq 1$,
$$\mu(\theta A+(1-\theta)B)\geq \mu(A)^{\theta}\mu(B)^{1-\theta}.$$
A random vector $Y$ is called \emph{log-concave} if its probability distribution is log-concave. We call $Y$ \emph{isotropic} if it is mean-zero and $\E\langle Y,x\rangle^2=\|x\|_2^2$ for all $x\in \R^n$. We say that $Y$ is \emph{unconditional} if, for any $\eps_1,\ldots,\eps_n \in \{-1,1\}$, the vector $(\eps_1Y_1,\ldots, \eps_nY_n)$ has the same distribution as $Y$. A typical example of an isotropic, unconditional log-concave vector $Y$ is a random variable uniformly distributed over the unit ball of an unconditional norm in the isotropic position.\par
We will use the following comparison theorem from \cite{MR2520452} (see also Theorem~2.5 in \cite{radeck_djalil}), which is based on earlier work in \cite{MR2083388}. It will allow us to reduce the general case of matrices with i.i.d.\ isotropic, unconditional, log-concave rows to the special case of a standard symmetric exponential matrix. 
\begin{theorem}\label{theo:bn}
Let $Y$ be an isotropic, unconditional, log-concave vector in $\R^d$ and $E$ be a standard $d$-dimensional symmetric exponential vector, i.e., its entries are i.i.d.\ standard symmetric exponential variables. Let $\norm{\cdot}$ be any semi-norm on $\R^d$. Then for any $t>0$,
\begin{equation*}
\Prob\big[\norm{Y}\geq Ct\big]\leq C\Prob\big[\norm{E}\geq t\big],
\end{equation*}where $C$ is a universal constant.
\end{theorem}
\begin{corollary}
\label{cor:logconc}
Let $A$ be an $m\ti n$ matrix with i.i.d.\ rows $X_i$ distributed as $X$, where $X$ is an isotropic, unconditional log-concave vector. If $m\gtrsim s^{2-2/q}\log(en/s) + \log(\eta^{-1})$, then the conclusion of Theorem~\ref{thm:main} holds. 
\end{corollary}
\begin{proof}
We verify the conditions of Theorem~\ref{thm:genDistributions}. By a result of Borell (see e.g.\ \cite[Proposition 2.14]{Led01}), $X$ is a sub-exponential vector. In fact, for any $p\geq 1$,
$$(\E|\langle X,x\rangle|^p)^{1/p} \lesssim p\E|\langle X,x\rangle| \qquad \operatorname{for \ all \ } x\in \R^n.$$
Since $X$ is isotropic, we can apply (\ref{eqn:P-Z}) for $|\langle X,x\rangle|^2$ to get
$$\bP(|\langle X,x\rangle|>u)\geq \frac{(\E|\langle X,x\rangle|^2 - u^2)^2}{\E|\langle X,x\rangle|^4} \gtrsim (1-u^2)^2,$$
whenever $0\leq u\leq 1$ and $\|x\|_2=1$. This shows that (\ref{eqn:SBass}) holds with absolute constants $u_*,\beta>0$.\par
To prove the second condition, we define a semi-norm on $\R^{m\times n}$ by
\begin{equation*}
\norm{B}_s=\sup_{x\in \Si_s^2}\Big\langle \sum_{i=1}^m B_{i},x\Big\rangle 
\end{equation*}
where the $B_i$ are the $m$ row vectors of $B\in \R^{m\times n}$. Since the $X_i$ are unconditional,
$$\E\sup_{x\in \Si_s^2}\langle V,x\rangle = \frac{1}{\sqrt{m}}\E\norm{A}_s.$$
Considered as a vector in $\R^{mn}$, $A$ is isotropic, unconditional and log-concave. Theorem~\ref{theo:bn} therefore implies that, 
\begin{equation*}
  \Prob\big[\norm{A}_s\geq C t\big]\leq C \Prob\big[\norm{\mathbb{\mathcal{E}}}_s\geq  t\big], 
\end{equation*}
where $\mathcal{E}$ is an $m\times n$ standard symmetric exponential matrix. As a consequence, we have
\begin{equation*}
\E\norm{A}_s = \int_0^\infty \Prob\big[\norm{A}_s\geq t\big] dt \leq C^2 \int_0^\infty \Prob\big[\norm{\cE}_s\geq  t\big] dt
  \lesssim \E\norm{\cE}_s. 
\end{equation*}
By the proof of Corollary~\ref{cor:condEntries} (see also (ii) of Example~\ref{exa:iid}), $\E\norm{\cE}_s\lesssim \sqrt{ms\log(en/s)}$, which proves the second condition in Theorem~\ref{thm:genDistributions}.
\end{proof}
As was mentioned before, Koltchinskii showed that $m\gtrsim s\log(en/s)$ isotropic, log-concave measurements suffice with high probability to recover every $s$-sparse vector exactly via $\ell_1$-minimization \cite[Theorem~7.3]{Kol11}. Under the additional assumption that the measurement vectors are unconditional, Corollary~\ref{cor:logconc} makes this result stable with respect to approximate sparsity and robust with respect to measurement noise, while retaining the optimal number of measurements.

 \section{RIP RIP?}
 \label{sec:RIPRIP}

The classical RIP property, (RIP$_{2,2}$), played a major role in the theory of compressed sensing since \cite{cata05,CaT06}. It has proved to be an optimal tool to analyze standard basis pursuit denoising for subgaussian matrices. It has also been used to show that various other random matrices, including structured random matrices, allow for uniform sparse recovery via (BPDN$_2$) if one increases the number of measurements with additional logarithmic factors. Nevertheless, it is known that for certain ensembles (e.g.\ subexponential) this logarithmic increase can be avoided, establishing a gap between RIP and sparse recovery conditions.  

In this work we showed that this gap becomes much more pronounced when considering (BPDN$_p$) for $p\neq 2$. An analysis of this program via an RIP condition erroneously suggests that 1) the required optimal number of measurements for uniform sparse recovery may be much larger than in the case $p=2$, especially if $p>2$, and 2) that one may need to consider random measurements different from Gaussian to attain this optimal number. This begs the question: does this mean that researchers interested in sparse recovery should stop
considering restricted isometry properties? In this paper we showed that by proving a lower (RIP$_{p,q}$)-type of bound on an \emph{extension} of the set of sparse vectors (cf.\ (\ref{eqn:claim1})), one can prove an optimal recovery result for a large class of matrices, which do not satisfy (RIP$_{p,q}$) in the optimal measurement regime. Thus, it seems the gap between RIP-properties and sparse recovery conditions originates in the upper bound of the RIP ``for all $x\in\Sigma_s, \norm{Ax}_p\leq C \norm{x}_q$'' -- at least when considering convex optimization approaches for recovery.

To move towards a definitive answer of our question, it would be interesting to determine whether similar gaps occur between RIP-properties and sparse recovery conditions for other numerical methods. For example, there are several algorithms such as iterative hard thresholding and CoSamp for which convergence results are currently only known under the (classical) RIP. 

\section*{Acknowledgements}

H.~Rauhut acknowledges funding by the European Research Council through the Starting Grant StG 258926.


\begin{thebibliography}{10}

\bibitem{ALL11}
R.~Adamczak, R.~Lata{\l}a, A.~Litvak, A.~Pajor, and N.~Tomczak-Jaegermann.
\newblock Geometry of log-concave ensembles of random matrices and approximate
  reconstruction.
\newblock {\em C. R. Math. Acad. Sci. Paris}, 349(13-14):783--786, 2011.

\bibitem{MR2796091}
R.~Adamczak, A.~E. Litvak, A.~Pajor, and N.~Tomczak-Jaegermann.
\newblock Restricted isometry property of matrices with independent columns and
  neighborly polytopes by random sampling.
\newblock {\em Constr. Approx.}, 34(1):61--88, 2011.

\bibitem{AGR15}
Z.~Allen-Zhu, R.~Gelashvili, and I.~Razenshteyn.
\newblock The restricted isometry property for the general p-norms.
\newblock {\em arXiv:1407.2178}, 2014.

\bibitem{BGI08}
R.~Berinde, A.~Gilbert, P.~Indyk, H.~Karloff, and M.~Strauss.
\newblock Combining geometry and combinatorics: A unified approach to sparse
  signal recovery.
\newblock In {\em Communication, Control, and Computing, 2008 46th Annual
  Allerton Conference on}, pages 798--805, Sept 2008.

\bibitem{MR2083388}
S.~G. Bobkov and F.~L. Nazarov.
\newblock On convex bodies and log-concave probability measures with
  unconditional basis.
\newblock In {\em Geometric aspects of functional analysis}, volume 1807 of
  {\em Lecture Notes in Math.}, pages 53--69. Springer, Berlin, 2003.

\bibitem{BJK14}
P.~T. Boufounos, L.~Jacques, F.~Krahmer, and R.~Saab.
\newblock Quantization and compressive sensing.
\newblock {\em arXiv:1405.1194}, 2014.

\bibitem{carota06}
E.~J. {C}and{\`e}s, {J}., T.~{T}ao, and J.~K. {R}omberg.
\newblock {R}obust uncertainty principles: exact signal reconstruction from
  highly incomplete frequency information.
\newblock {\em {I}{E}{E}{E} {T}rans. {I}nform. {T}heory}, 52(2):489--509, 2006.

\bibitem{cata05}
E.~J. {C}and{\`e}s and T.~{T}ao.
\newblock {D}ecoding by linear programming.
\newblock {\em {I}{E}{E}{E} {T}rans. {I}nform. {T}heory}, 51(12):4203--4215,
  2005.

\bibitem{CaT06}
E.~J. Candes and T.~Tao.
\newblock Near-optimal signal recovery from random projections: universal
  encoding strategies?
\newblock {\em IEEE Trans. Inform. Theory}, 52(12):5406--5425, 2006.

\bibitem{radeck_djalil}
D.~Chafa{\"{i}} and R.~Adamczak.
\newblock Circular law for random matrices with unconditional log-concave
  distribution.
\newblock 2014.
\newblock arXiv:1303.5838.

\bibitem{Cha10}
V.~B. Chandar.
\newblock {\em Sparse graph codes for compression, sensing, and secrecy}.
\newblock PhD thesis, Massachusetts Institute of Technology, 2010.

\bibitem{do06-2}
D.~L. {D}onoho.
\newblock {C}ompressed sensing.
\newblock {\em {I}{E}{E}{E} {T}rans. {I}nform. {T}heory}, 52(4):1289--1306,
  2006.

\bibitem{MR3134335}
S.~Foucart.
\newblock Stability and robustness of {$\ell_1$}-minimizations with {W}eibull
  matrices and redundant dictionaries.
\newblock {\em Linear Algebra Appl.}, 441:4--21, 2014.

\bibitem{FoR13}
S.~Foucart and H.~Rauhut.
\newblock {\em A Mathematical Introduction to Compressive Sensing}.
\newblock Birkh{\"{a}}user, Boston, 2013.

\bibitem{Fuc09}
J.-J. Fuchs.
\newblock Fast implementation of a {$\ell_1$}-{$\ell_1$} regularized sparse
  representations algorithm.
\newblock In {\em Acoustics, Speech and Signal Processing, 2009. ICASSP 2009.
  IEEE International Conference on}, pages 3329--3332, April 2009.

\bibitem{JHF11}
L.~Jacques, D.~K. Hammond, and J.~M. Fadili.
\newblock Dequantizing compressed sensing: when oversampling and non-{G}aussian
  constraints combine.
\newblock {\em IEEE Trans. Inform. Theory}, 57(1):559--571, 2011.

\bibitem{KaR15}
M.~Kabanava and H.~Rauhut.
\newblock Analysis $\ell_1$-recovery with frames and {G}aussian measurements.
\newblock {\em Acta Appl. Math.}, to appear.
\newblock {DOI}:10.1007/s10440-014-9984-y.

\bibitem{Kol11}
V.~Koltchinskii.
\newblock {\em Oracle {I}nequalities in {E}mpirical {R}isk {M}inimization and
  {S}parse {R}ecovery {P}roblems}.
\newblock Springer, Berlin, 2011.

\bibitem{KoM13}
V.~Koltchinskii and S.~Mendelson.
\newblock Bounding the smallest singular value of a random matrix without
  concentration.
\newblock {\em Int. Math. Res. Notices}, to appear.
\newblock arXiv:1312.3580.

\bibitem{MR2520452}
R.~Lata{\l}a.
\newblock On weak tail domination of random vectors.
\newblock {\em Bull. Pol. Acad. Sci. Math.}, 57(1):75--80, 2009.

\bibitem{LeM14}
G.~Lecu{\'e} and S.~Mendelson.
\newblock Sparse recovery under weak moment assumptions.
\newblock {\em J. Eur. Math. Soc.}, to appear.
\newblock ArXiv:1401.2188.

\bibitem{Led01}
M.~Ledoux.
\newblock {\em The concentration of measure phenomenon}, volume~89 of {\em
  Mathematical Surveys and Monographs}.
\newblock American Mathematical Society, Providence, RI, 2001.

\bibitem{Sh14_2}
S.~Mendelson.
\newblock Learning without concentration for general loss functions.
\newblock {\em arXiv:1410.3192}, 2014.

\bibitem{Sh14_3}
S.~Mendelson.
\newblock A remark on the diameter of random sections of convex bodies.
\newblock In B.~Klartag and E.~Milman, editors, {\em Geometric Aspects of
  Functional Analysis (GAFA Seminar Notes)}, volume 2116 of {\em Lecture notes
  in Mathematics}, pages 395--404, 2014.

\bibitem{Sh14_1}
S.~Mendelson.
\newblock Learning without concentration.
\newblock {\em J. ACM}, to appear.
\newblock arXiv:1401.0304.

\bibitem{MR2373017}
S.~Mendelson, A.~Pajor, and N.~Tomczak-Jaegermann.
\newblock Reconstruction and subgaussian operators in asymptotic geometric
  analysis.
\newblock {\em Geom. Funct. Anal.}, 17(4):1248--1282, 2007.

\bibitem{Nat95}
B.~K. Natarajan.
\newblock Sparse approximate solutions to linear systems.
\newblock {\em SIAM Journal on Computing}, 24(2):227--234, 1995.

\bibitem{pi85}
A.~{P}inkus.
\newblock {\em $n$-{W}idths in approximation theory}.
\newblock {S}pringer-{V}erlag, 1985.

\bibitem{RuV08}
M.~Rudelson and R.~Vershynin.
\newblock On sparse reconstruction from {F}ourier and {G}aussian measurements.
\newblock {\em Comm. Pure Appl. Math.}, 61(8):1025--1045, 2008.

\end{thebibliography}

\end{document}